\documentclass[a4paper,UKenglish,cleveref,autoref, thm-restate]{lipics-v2021}
\usepackage{algorithm}
\usepackage{algorithmic}
\usepackage{xcolor}
\usepackage{float}
\usepackage{graphicx}

\def\hf{\selectfont\sffamily\bfseries}
\renewcommand{\paragraph}[1]{\smallskip\noindent{\hf #1}}

\def\DSSC{\mathrm{Mult}\text{-}\mathrm{MSSC}}
\def\SSC{\mathrm{MSSC}}
\def\cR{\mathcal{R}}

\def\dkt{\mathrm{d}_{\mathrm{KT}}}
\def\dfr{\mathrm{d}_{\mathrm{FR}}}
\title{On the Approximability of Multistage Min-Sum Set Cover} 

\titlerunning{On the Approximability of Multistage Min-Sum Set Cover} 

\author{Dimitris Fotakis}{National Technical University of Athens}{fotakis@cs.ntua.gr}{[0000-0001-6864-8960}{Supported by the Hellenic Foundation for Research and Innovation (H.F.R.I.) under the ``First Call for H.F.R.I. Research Projects to support Faculty members and Researchers and the procurement of high-cost research equipment grant'',  project BALSAM, HFRI-FM17-1424.}

\author{Panagiotis Kostopanagiotis}{National Technical University of Athens}{panagiotis.kostopanagiotis@gmail.com}{}{}

\author{Vasileios Nakos}{Saarland University and Max Planck Institute for Informatics}{vnakos@mpi-inf.mpg.de}{}{ Supported by the project TIPEA that has
received funding from the European Research Council (ERC) under the European Unions Horizon
2020 research and innovation programme (grant agreement No. 850979).}

\author{Georgios Piliouras}{Singapore University of Technology and Design}{georgios.piliouras@gmail.com}{}{Supported by NRF2019-NRF-ANR095 ALIAS grant, grant
PIE-SGP-AI-2018-01, NRF 2018 Fellowship NRF-NRFF2018-07,  AME Programmatic Fund (Grant No. A20H6b0151) from the Agency for Science, Technology and Research (A*STAR) and AI Singapore grant AISG2-RP-2020-016.}

\author{Stratis Skoulakis}{Singapore University of Technology and Design}{efstratios@sutd.edu.sg}{}{ Supported by NRF 2018 Fellowship NRF-NRFF2018-07.}

\authorrunning{D. Fotakis, P. Kostopanagiotis, V. Nakos, G. Piliouras and S. Skoulakis}

\Copyright{Dimitris Fotakis, Panagiotis Kostopanagiotis, Vasilis Nakos, Georgios Piliouras and Stratis Skoulakis} 



\category{Track A: Algorithms, Complexity and Games}

\relatedversion{} 

\nolinenumbers 

\EventEditors{Nikhil Bansal, Emanuela Merelli, and James Worrell}
\EventNoEds{3}
\EventLongTitle{48th International Colloquium on Automata, Languages, and Programming (ICALP 2021)}
\EventShortTitle{ICALP 2021}
\EventAcronym{ICALP}
\EventYear{2021}
\EventDate{July 12--16, 2021}
\EventLocation{Glasgow, Scotland (Virtual Conference)}
\EventLogo{}
\SeriesVolume{198}
\ArticleNo{108}

\begin{document}

\maketitle

\begin{abstract}
We investigate the polynomial-time approximability of the multistage version of Min-Sum Set Cover ($\DSSC$), a natural and intriguing generalization of the classical List Update problem. In $\DSSC$, we maintain a sequence of permutations $(\pi^0, \pi^1, \ldots, \pi^T)$ on $n$ elements, based on a sequence of requests $\cR = (R^1, \ldots, R^T)$. We aim to minimize the total cost of updating $\pi^{t-1}$ to $\pi^{t}$, quantified by the Kendall tau distance $\dkt(\pi^{t-1}, \pi^t)$, plus the total cost of covering each request $R^t$ with the current permutation $\pi^t$, quantified by the position of the first element of $R^t$ in $\pi^t$. 

Using a reduction from Set Cover, we show that $\DSSC$ does not admit an $O(1)$-approximation, unless $\mathrm{P} = \mathrm{NP}$, and that any $o(\log n)$ (resp. $o(r)$) approximation to $\DSSC$ implies a sublogarithmic (resp. $o(r)$) approximation to Set Cover (resp. where each element appears at most $r$ times). Our main technical contribution is to show that $\DSSC$ can be approximated in polynomial-time within a factor of $O(\log^2 n)$ in general instances, by randomized rounding, and within a factor of $O(r^2)$, if all requests have cardinality at most $r$, by deterministic rounding. 
\end{abstract}

\section{Introduction}
\label{s:intro}

In \emph{Multistage Min-Sum Set Cover} ($\DSSC$), we are given a universe $U$ on $n$ elements, a sequence of requests $\cR = (R_1, \ldots, R_T)$, with $R_t \subseteq U$, and an initial permutation $\pi^0$ of the elements of $U$. We aim to maintain a sequence of permutations $(\pi^0, \pi^1, \ldots, \pi^T)$ of $U$, so as to minimize the total cost of updating (or moving from) $\pi^{t-1}$ to $\pi^{t}$ in each time step plus the total cost of covering each request $R_t$ with the current permutation $\pi^t$. The cost of moving from $\pi^{t-1}$ to $\pi^{t}$ is the number of inverted element pairs between $\pi^{t-1}$ and $\pi^t$, i.e., the Kendall Tau distance $\dkt(\pi^{t-1}, \pi^t)$. The cost $\pi^t(R_t)$ of covering a request $R_t$ with a permutation $\pi^t$ is the position of the first element of $R_t$ in $\pi^t$, i.e., $\pi^t(R_t) = \min\{ i \,|\, \pi^t(i) \in R_t \}$. Thus, given $\cR = (R_1, \ldots, R_T)$, we aim to minimize $\sum_{t=1}^T \big( \dkt(\pi^{t-1}, \pi^t) + \pi^t(R_t) \big)$. 
%

The $\DSSC$ problem is a natural generalization of the (offline version of the) classical \emph{List Update} problem \cite{ST85b}, where $|R_t| = 1$ for all requests $R_t \in \cR$. The offline version of List Update is $\mathrm{NP}$-hard \cite{Amb00}, while it is known that any $5/4$-approximation has to resort to \emph{paid exchanges}, where an element different from the requested one is moved forward to the list \cite{LRR15,tim16}. $\DSSC$ was introduced in \cite{FKKSV20} as the multistage extension of Min-Sum Set Cover ($\SSC$) \cite{FLT04}, where we aim to compute a single static permutation $\pi$ that minimizes the total covering cost  $\sum_{t=1}^T \pi(R_t)$. \cite{FKKSV20} presented a (simple polynomial-time) online algorithm for $\DSSC$ with competitive ratio between $\Omega(r \sqrt{n})$ and $O(r^{3/2} \sqrt{n})$ for $r$-bounded instances, where all requests have cardinality at most $r$, and posed the polynomial-time approximability of $\DSSC$ as an interesting open question. $\DSSC$ is also related to recently studied time-evolving (a.k.a. multistage or dynamic) optimization problems (e.g., multistage matroid, spanning set and perfect matching maintenance \cite{GTW14}, time-evolving Facility Location \cite{EMS14,Svensson15}), where we aim to maintain a sequence of near-optimal feasible solutions to a combinatorial optimization problem, in response to time-evolving underlying costs, without changing too much the solution from one step to the next. 

\paragraph{Motivation.} 
$\DSSC$ is motivated by applications, such as web search, news, online shopping, paper bidding, etc., where items are presented to the users sequentially. Then, the item ranking is of paramount importance, because user attention is usually restricted to the first few items in the sequence (see e.g., \cite{StreeterGK09,CKMS01,FSR18,PT18}). If a user does not spot an item fitting her interests there, she either leaves the service (in case of news or online shopping, see e.g., the empirical evidence in \cite{DGMM20}) or settles on a suboptimal action (in case of paper bidding, see e.g., \cite{GP13}). To mitigate such situations and increase user retention, modern online services highly optimize item rankings based on user scrolling and click patterns. Each user $t$ is represented by her set of preferred items (or item categories) $R_t$\,. The goal of the service provider is to continually maintain an item ranking $\pi^t$, so that the current user $t$ finds one of her favorite items at a relatively high position in $\pi^t$. Continual ranking update is dictated by the fact that users with different characteristics and preferences tend to use the online service during the course of the day (e.g., elderly people in the morning, middle-aged people in the evening, young people at the night -- similar patterns apply for people from different countries and timezones). Moreover, different user categories react in nonuniform ways to different trends (in e.g., news, fashion, sports, scientific topics). For consistency and stability, however, the ranking should change neither too much nor too frequently. $\DSSC$ makes the (somewhat simplifying) assumptions that the service provider has a relatively accurate knowledge of user preferences and their arrival order, and that its total cost is proportional to how deep in $\pi^t$ the current user $t$ should reach, before she finds one of her favorite items, and to how much the ranking changes from one user to the next. 

From a theoretical viewpoint, $\DSSC$ was used in \cite{FKKSV20} as a natural benchmark for studying the dynamic competitive ratio of Online Min-Sum Set Cover, where the algorithm updates its permutation online, without any knowledge of future requests. As in $\DSSC$, the objective is to minimize the total moving plus the total covering cost. 

\paragraph{Contribution and Techniques.}
In this work, we initiate a study of the polynomial-time approximability of $\DSSC$. Using a reduction from Set Cover, we show (Theorem~\ref{t:hardness}) that $\DSSC$ does not admit a $c\log n$-approximation, for some absolute constant $c$, unless $\mathrm{P}=\mathrm{NP}$. Moreover our reduction establishes that an $o(r)$-approximation for $r$-bounded instances of $\DSSC$ implies an $o(r)$-approximation for Set Cover, in case each element appears in at most $r$ requests. 

Our main technical contribution is to show that $\DSSC$ can be approximated in polynomial-time within a factor of $O(\log^2 n)$ in general instances, by randomized rounding (Theorem~\ref{t:rand}), and within a factor of $O(r^2)$ in $r$-bounded instances, by deterministic rounding (Theorem~\ref{t:greedy}). 

For both results, we consider a restricted version of $\DSSC$, inspired by the Move-to-Front (MTF) algorithm for List Update, where in each time step $t$, we can only move a single element of $R_t$ from its position in $\pi^{t-1}$ to the first position of $\pi^t$. Since such a permutation $\pi^t$ coves $R_t$ with unit cost, we now aim to select the element of each $R_t$ moved to front of $\pi^t$, so as to minimize the total moving cost $\sum_{t=1}^T \dkt(\pi^{t-1}, \pi^t)$. It is not hard to see that the optimal cost of serving $\cR$ under the restricted Move-to-Front version of $\DSSC$ is within a factor of $4$ from the optimal cost under the original, more general, definition of $\DSSC$. 

Hence, approximating $\DSSC$ boils down to determining which element of $R_t$ should become the top element of $\pi^{t}$. To this end, we relax permutations to doubly stochastic matrices and consider a Linear Programming relaxation of the restricted Move-to-Front version of $\DSSC$, which we call \emph{Fractional-MTF} (see Definition~\ref{d:frac_MTF}). Given the optimal solution of the aforementioned linear program, which is a sequence of doubly stochastic matrices $(A^0, A^1, \ldots, A^T)$, with $A^0$ corresponding to the initial permutation $\pi^0$, our main technical challenge is to round each doubly stochastic matrix $A^t$ to a permutation $\pi^t$ such that (i) there is an element of $R_t$ at one of the few top positions of $\pi^t$; and (ii) the total moving cost $\sum_{t=1}^T \dkt(\pi^{t-1}, \pi^t)$ of the rounded solution is comparable to the total moving cost $\sum_{t=1}^T \dfr(A^{t-1}, A^t)$ of the optimal solution of Fractional-MTF, where $\dfr$ is a notion of distance equivalent to Spearman's footrule distance on permutations (see Definition~\ref{d:distance_lp}).  


Working towards a randomized rounding approach, we first observe that rounding each doubly stochastic matrix independently may result in a permutation sequence with total moving cost significantly larger than that of Fractional-MTF (see also the discussion after Lemma~\ref{l:relax}). In Theorem~\ref{t:rand}, we show that a dependent randomized rounding with  logarithmic scaling of entries  (Algorithm~\ref{alg:rand_rounding}), similar in spirit with the randomized rounding approach \cite{BGK10,SW11} for Generalized Min-Sum Set Cover, results in an approximation ratio of $O(\log^2 n)$. Interestingly, Algorithm~\ref{alg:rand_rounding} without the logarithmic scaling results in a permutation sequence with the expected moving cost within a factor of $4$ from the optimal moving cost of Fractional-MTF. However, we lose a logarithmic factor in the approximation ratio, because we need to scale up the entries of each doubly stochastic matrix $A^t$, so as to ensure that some element of $R_t$ appears in the few top positions of $\pi^t$ with sufficiently large probability. The other logarithmic factor is lost because there could be a logarithmic number of elements allocated to the same position of the resulting permutation by the randomized rounding. 

Our deterministic rounding of Algorithm~\ref{alg:greedy_rounding} for $r$-bounded request sequences is motivated by the deterministic rounding for Set Cover and Vertex Cover. We observe that in the optimal solution of Fractional-MTF, in each time step $t$, there is some element $e \in R_t$ with $A^{t}_{e1} \geq 1/r$ (i.e., $e$ occupies a fraction of at least $1/r$ of the first position in the ``fractional permutation'' $A^t$). Algorithm~\ref{alg:greedy_rounding} simply moves any such element to the front of $\pi^t$. The most challenging part of the analysis is to establish that for any optimal solution $(A^0, A^1, \ldots, A^T)$ of Fractional-MTF with respect to an $r$-bounded request sequence, there exists a sequence of doubly stochastic matrices $(A^0, \hat{A}^1, \ldots, \hat{A}^T)$ with the entries of each $\hat{A}^t$ being multiples of $1/r$, such that (i) the moving cost of $(A^0, \hat{A}^1, \ldots, \hat{A}^T)$ is bounded from above by the optimal cost of Fractional-MTF; and (ii) each matrix $\hat{A}^t$ contains in the first position the element that Algorithm~\ref{alg:greedy_rounding} keeps in the first position at round $t$, with mass at least $1/r$. Then we show (Lemma~\ref{l:r_integral}) that for any sequence of doubly stochastic matrices $(A^0, \hat{A}^1, \ldots, \hat{A}^T)$ satisfying the above properties, the moving cost of Algorithm~\ref{alg:greedy_rounding} is at most the moving cost of the doubly stochastic matrices,
$\sum_{t=1}^T \dfr(\hat{A}^t,\hat{A}^{t-1})$. The latter is done through the use of an appropriate potential function based on an extension of the Kendall-Tau distance to doubly stochastic matrix with entries being multiples of $1/r$.

A potentially interesting insight 
is that the technical reason for the quadratic dependence of our approximation ratios on $\log n$ and $r$ is conceptually similar to the reason for the (best possible) approximation ratio of $4 = 2\cdot 2$ in \cite{FLT04} (see the discussion after Theorem~\ref{t:rand}). Hence, we conjecture that any $o(\log^2 n)$ (resp. $o(r^2)$) approximation to $\DSSC$ must imply a sublogarithmic (resp. $o(r)$) approximation to Set Cover.


\paragraph{Other Related Work.}
The $\SSC$ problem generalizes various $\mathrm{NP}$-hard problems, such as Min-Sum Vertex Cover and Min-Sum Coloring and it is well-studied. Feige, Lovasz and Tetali~\cite{FLT04} proved that the greedy algorithm, which picks in each position the element that covers the most uncovered requests, is a $4$-approximation (that was also implicit in~\cite{BBHST98}) and that no $(4-\varepsilon)$-approximation is possible, unless $\mathrm{P} = \mathrm{NP}$. In Generalized $\SSC$ (a.k.a. \emph{Multiple Intents Re-ranking}), there is a covering requirement $K(R_t)$ for each request $R_t$ and the cost of covering a request $R_t$ is the position of the $K(R_t)$-th element of $R_t$ in the (static) permutation $\pi$. The $\SSC$ problem is the special case where $K(R_t)=1$ for all requests $R_t$. Another notable special case of Generalized $\SSC$ is the Min-Latency Set Cover problem \cite{HL05}, which corresponds to the other extreme case where $K(R_t) = |R_t|$ for all requests $R_t$. Generalized $\SSC$ was first studied by Azar et al.~\cite{AGY09}, who presented a $O(\log r)$-approximation; later $O(1)$-approximation algorithms were obtained~\cite{BGK10,SW11,ISZ14,BBFT20}.

Further generalizations of Generalized $\SSC$ have been considered, such as the Submodular Ranking problem, studied in \cite{AG11}, which generalizes both Set Cover and $\SSC$, and the Min-Latency Submodular Cover, studied by Im et al.~\cite{INZ16}. We refer to~\cite{INZ16,Im16} for a detailed discussion on the connections between these problems and their applications. 

The online version of $\SSC$, which generalizes the famous List Update problem, was studied in \cite{FKKSV20}. They proved that its static deterministic competitive ratio is $\Theta(r)$ and presented a natural  memoryless algorithm, called \emph{Move-all-Equally}, with static competitive ratio in $\Omega(r^2)$ and $2^{O(\sqrt{\log n  \cdot \log r})}$ and dynamic competitive ratio in $\Omega(r \sqrt{n})$ and $O(r^{3/2} \sqrt{n})$-competitive. Subsequently, \cite{FLPS20} considered $\SSC$ from the viewpoint of online learning. Through dimensionality reduction from permutations to doubly stochastic matrices, they obtained randomized (resp. deterministic) polynomial-time online learning algorithms with $O(1)$-regret for Generalized $\SSC$ (resp. $O(r)$-regret for $\SSC$). 


\section{Preliminaries and Basic Definitions}
The set of elements $e$ is denoted by $U$ with $|U| = n$. A permutation of the elements is denoted by $\pi$ where $\pi_i$ denotes the element lying at position $i$ (for $1\leq i \leq n$) and $\mathrm{Pos}(e,\pi)$ denotes the position of the element $e \in U$ in permutation $\pi$.

\begin{definition}[Kendall-Tau Distance]\label{d:KT}
Given the permutations $\pi^A,\pi^B$, a pair of elements $(e,e')$ is inverted if and only if $\mathrm{Pos}(e,\pi^A) > \mathrm{Pos}(e',\pi^A)$ and $\mathrm{Pos}(e,\pi^B) < \mathrm{Pos}(e',\pi^B)$ or vice versa. The Kendall-Tau distance between the permutations $\pi^A,\pi^B$, denoted by $\dkt(\pi^A,\pi^B)$, is the number of inverted pairs.
\end{definition}

\begin{definition}[Spearman' Footrule Distance]\label{d:FR}
The FootRule distance between the permutations $\pi^A,\pi^B$ is defined as $\mathrm{d}_{\mathrm{FR}}(\pi^A,\pi^B) = \sum_{e \in U}|\mathrm{Pos}(e,\pi^A) - \mathrm{Pos}(e,\pi^B)|$.
\end{definition}
\noindent The Kendall-Tau distance and FootRule distance are approximately equivalent, $\dkt(\pi^A,\pi^B) \leq \mathrm{d}_{\mathrm{FR}}(\pi^A,\pi^B) \leq 2 \cdot \dkt(\pi^A,\pi^B)$. Moreover both of them satisfy the triangle inequality.
\begin{definition} An $n \times n$ matrix with positive entries (rows stand for the elements and columns for the positions) is called stochastic if $\sum_{i = 1}^n A_{ei} = 1$ for all $e\in U$ and doubly stochastic if (additionally) $\sum_{e \in U} A_{ei}=1$ for all $1\leq i \leq n$.
\end{definition}
\noindent A permutation of the elements $\pi$ can be equivalent represented by a $0$-$1$ doubly stochastic matrix $A$, where $A_{ei}=1$ if element $e$ lies at position $i$ and $0$ otherwise. When clear from context, we use the notion of permutation and ($0$-$1$) doubly stochastic matrix interchangeably.

The notion of FootRule distance can be naturally extended to stochastic matrices. Given two doubly stochastic matrices $A,B$ consider the min-cost transportation problem, transforming row $A_e$ to the row $B_e$ where the cost of transporting a unit of mass between column $i$ and column $j$ equals $|i-j|$. Formally for each row $e$, define a complete bipartite graph where on the left part lie the entries $(e,i)$ for $1\leq i \leq n$ and on the right part the entries $(e,j)$ for $1\leq j \leq n$. The mass transported from entry $(e,i)$ to entry $(e,j)$ (denoted as $f_{ij}^e$)
costs $f_{ij}^e\cdot |i-j|$ and the total mass \textit{leaving} $(e,i)$ equals $A_{ei}$ and the total  mass \textit{arriving} at $(e,j)$ equals $B_{ej}$.
\begin{definition}\label{d:distance_lp}
The FootRule distance between two stochastic matrices 
$A,B$, denoted by $\mathrm{d}_{\mathrm{FR}}(A,B)$, is the optimal value of the following linear program,
\begin{equation*}
    \begin{array}{ll@{}ll}
        \text min & \displaystyle \sum_{e \in U} \sum_{i = 1}^{n} \sum_{j=1}^n |i - j| \cdot f_{ij}^e &\\
        \text{ s.t } & \displaystyle \sum_{i=1}^{n} f_{ij}^e = B_{ej}~~~~\text{for all }e \in U \text{ and }j=1,\ldots ,n&&\\
        & \displaystyle \sum_{j=1}^{n} f_{ij}^e = A_{ei}~~~~\text{for all }e \in U \text{ and }i=1,\ldots, n&&\\
        & \displaystyle f_{ij}^e \geq 0~~~~~~~~~~~\text{for all } e \in U \text{ and }i,j = 1,\ldots, n
    \end{array}
\end{equation*}
\end{definition}
\begin{example}
Let the stochastic matrices $A = 
\begin{pmatrix}
1 & 0 & 0 \\
0 & 1 & 0 \\
0 & 0 & 1
\end{pmatrix}$, $B = 
\begin{pmatrix}
1/3 & 1/3 & 1/3 \\
1/2 & 1/2 & 0 \\
1/4 & 0 & 3/4
\end{pmatrix}$. The FootRule distance $\mathrm{d}_{\mathrm{FR}}(A,B) = \underbrace{(0\cdot 1/3 + 1\cdot 1/3 + 2\cdot 1/3)}_{\text{first row}}$ + $\underbrace{(1\cdot 1/2 + 0\cdot 1/2 + 1\cdot 0)}_{\text{second row}}$
+ $\underbrace{(2\cdot 1/4 + 1\cdot 0 + 0\cdot 3/4)}_{\text{third row}} = 2$.
\end{example}
\noindent Up next we present the formal definition of Multistage Min-Sum Set Cover.
\begin{definition}[\textbf{Multistage Min-Sum Set Cover}]
Given a universe of elements $U$, a sequence of
requests $R_1,\ldots,R_T \subseteq U$ and an initial permutation of the elements $\pi^0$. 
The goal is to select a sequence of permutation $\pi^1,\ldots,\pi^T$ that minimizes 
$$\sum_{t=1}^{T} \pi^t(R_t) + \sum_{t=1}^{T} \dkt ( \pi^t, \pi^{t-1} )$$
where $\pi^t(R_t)$ is the position of the first element of $R_t$ that we encounter in $\pi^t$, $\pi^t(R_t) = \min \{1\leq i \leq n:~ \pi_i^t \in R_t\}$.
\end{definition}
\noindent We refer to $\sum_{t=1}^T\pi^t(R_t)$ as \textbf{covering cost} and to $\sum_{t=1}^T \dkt(\pi^t,\pi^{t-1})$ as \textbf{moving cost}.
We denote with $\pi_{\mathrm{Opt}}^t$ the permutation of the optimal solution of $\DSSC$ at round $t$, with $o_t$ the element that the optimal solution uses to cover the request $R_t$ (the element of $R_t$ appearing first in $\pi_{\mathrm{Opt}}^t$), and with $\mathrm{OPT}_{\DSSC}$ the cost of the optimal solution. Finally we call an instance of $\DSSC$ \textit{r-bounded} in case the cardinality of the requests is bounded by $r$, $|R_t| \leq r$.
\section{Approximation Algorithms for Multistage Min-Sum Set Cover}\label{s:main}
\noindent There exists an
approximation-preserving reduction from $\mathrm{Set-Cover}$ to  
$\DSSC$ that provides us with the following inapproximability results.
\begin{theorem}\label{t:hardness}
\begin{itemize}
    \item There is no $c \cdot \log n$-approximation algorithm for $\DSSC$ (for a sufficienly small constant $c$) unless $\mathrm{P = NP}$.
    
    \item For $r$-bounded sequences, there is no $o(r)$-approximation algorithm for $\DSSC$, unless there is a
    $o(r)$-approximation algorithm for $\mathrm{Set-Cover}$ with each element being covered by at most $r$ sets.
\end{itemize}
\end{theorem}

\noindent The proof of Theorem~\ref{t:hardness} is fairly simple, given an instance of $\mathrm{Set-Cover}$ we construct an instance of $\DSSC$ in which the initial permutation $\pi^0$ contains in the first positions some dummy elements (they do not appear in any of the requests) and in the last positions the sets of the $\mathrm{Set-Cover}$ (we consider an element of $\DSSC$ for each set of $\mathrm{Set-Cover}$). Finally each request for $\DSSC$ is associated with an element of the $\mathrm{Set-Cover}$ and contains the \textit{elements in $\DSSC$/ sets in $\mathrm{Set-Cover}$} containing it.

\begin{proof}
Let the equivalent definition of $\mathrm{Set-Cover}$ in which we are given a universe of element $E = \{1,\ldots,n\}$ and sets $S_1,S_2,\ldots,S_m \subseteq E$ and we are asked to select the minimum number of elements covering all the sets (an element $e$ covers set $S_i$ if $e \in S_i$).

Consider the instance of $\DSSC$ with the elements $U= \{1,\ldots,n\} \cup \{d_1,\ldots,d_{n^2m}\}$. The elements $\{d_1,\ldots,d_{n^2m}\}$
are dummy in the sense that
they appear in none of the requests $R_t$. Let the initial permutation $\pi_0$ contain in the first $n^2m$ positions the dummy elements and in the last $n$ positions the elements $\{1,\ldots,n\}$, $\pi_0 = [d_1,\ldots,d_{n^2m},1,\ldots,n]$ and the request sequence of 
$\DSSC$ be $S_1,S_2,\ldots,S_m$.\\

Let a $c$-approximation algorithm for $\DSSC$ producing the permutation $\pi_1,\ldots,\pi_m$ the cost of which is denoted by $\mathrm{Alg}$. Let also $\mathrm{CoverAlg}$ denote the set composed by the element that the $c$-approximation algorithm uses to cover the requests, $\mathrm{CoverAlg} = \{\text{the element of }S_t\text{ appearing first in }\pi_t\}$. Then,
$$\mathrm{Alg} \geq n^2 m \cdot |\mathrm{CoverAlg}|$$

Now consider the following solution 
for $\DSSC$ constructed by the optimal solution for $\mathrm{Set-Cover}$. This solution initially moves the elements of the optimal covering set
$\mathrm{OPT}_\mathrm{SetCover}$ to the first positions and then never changes the permutation. Clearly the cost of this solution is upper bounded by
$$\mathrm{Set-Cover}_{\DSSC} \leq \underbrace{|\mathrm{OPT}_\mathrm{SetCover}| \cdot (n^2m + n)}_{\text{moving cost}} + \underbrace{m\cdot |\mathrm{OPT}_\mathrm{SetCover}|}_{\text{covering cost}}  $$
\noindent In case $\mathrm{Alg} \leq c \cdot \mathrm{Set-Cover}_{\DSSC}$, we directly get that $|\mathrm{CoverAlg}| \leq 3c \cdot |\mathrm{OPT}_\mathrm{SetCover}|$.\\

\noindent There is no polynomial-time approximation algorithm for $\mathrm{Set}-\mathrm{Cover}$ with approximation ratio better than $\log m$. The latter holds even for instance of $\mathrm{Set}-\mathrm{Cover}$ for which $m = \mathrm{poly}(n)$ \cite{alon2006} where $\mathrm{poly}(\cdot)$ is a polynomial with degree bounded by a universal constant. Since the number of elements $|U|$, in the constructed instance of $\DSSC$ is $n^2m$, any $c\cdot \log |U|$-approximation for $\DSSC$ (for $c$ sufficiently small) implies an approximation algorithm for $\mathrm{Set}-\mathrm{Cover}$
with approximation ratio less than $\log n$. In case
there exists an $c=o(r)$-approximation algorithm for $\DSSC$ for requests sequences $R_1,\ldots,R_T$ where $|R_t|\leq r$, we obtain an $o(r)$-approximation for algorithm for $\mathrm{Set-Cover}$ for sets with cardinality bounded by $r$. In the standard form of $\mathrm{Set}-\mathrm{Cover}$ this is translated into the fact that each element belongs in at most $r$ sets.
\end{proof}

Both the $O(\log^2 n)$-approximation algorithm (for requests of general cardinality) and the $O(r^2 )$-approximation algorithm for $r$-bounded requests, that we subsequently present,
are based on rounding a linear program called \textit{Fractional Move To Front}. The latter is the linear program relaxation of \textit{Move To Front}, a problem closely related to Multistage Min-Sum Set Cover.$~\mathrm{MTF}$ asks for a sequence of permutations $\pi^1,\ldots,\pi^T$ such as at each round $t$, an element of $R_t$ lies on the first position of $\pi^t$ and $\sum_{t=1}^T \mathrm{d}_{\mathrm{FR}}(\pi^t,\pi^{t-1})$ is minimized.

\begin{definition}\label{d:frac_MTF}
Given a sequence of requests $R_1,\ldots,R_T \subseteq U$ and an initial permutation of the elements $\pi^0$, consider the following linear program, called $\mathrm{Fractional- MTF}$,
\begin{equation*}
    \begin{array}{ll@{}ll}
        \text min & \displaystyle \sum_{t=1}^{T} \mathrm{d}_{\mathrm{FR}}(A^t,A^{t-1})&\\
        \text{ s.t } & \displaystyle \sum_{i=1}^{n} A_{ei}^t = 1 ~~~~\text{for all }e \in U\text{ and } t = 1,\ldots, T &&\\
        & \displaystyle \sum_{e \in U} A_{ei}^t = 1 ~~~~\text{for all }i = 1,\ldots,n\text{ and } t = 1,\ldots, T&\\
        & \displaystyle \sum_{e \in R_t} A_{e1}^t = 1 ~~~\text{for all } t = 1,\ldots, T&\\
        & \displaystyle A^0 = \pi^0 &&\\
        & \displaystyle A_{ei}^t \geq 0 ~~~~~~~~\text{for all } e \in U, ~i = 1,\ldots, n\text{ and } t = 1,\ldots, T&\\
    \end{array}
\end{equation*}
where $\mathrm{d}_{\mathrm{FR}}(\cdot,\cdot)$ is the FootRule distance of Definition~\ref{d:distance_lp}.
\end{definition}
\noindent  There is an elegant argument (appeared in previous works, e.g., \cite{FKKSV20})
showing that the optimal solution of $\mathrm{MTF}$ is at most $4 \cdot \mathrm{OPT}_{\DSSC}$. In Lemma~\ref{l:relax} we provide the argument and establish that $\mathrm{Fractional- MoveToFront}$ is a $4$-approximate relaxation of $\DSSC$.
\begin{lemma}\label{l:relax}
$\sum_{t=1}^T \mathrm{d}_{\mathrm{FR}}(A^t,A^{t-1}) \leq 4 \cdot \mathrm{OPT}_{\DSSC}$
where $A^1,\ldots,A^t$ is the optimal solution of $\mathrm{Fractional- MTF}$.
\end{lemma}

\begin{proof}[Proof of Lemma~\ref{l:relax}]
Let $o_t$ the element of $R_t$ appearing first in the permutation $\pi_{\mathrm{Opt}}^t$. Consider the sequence of permutation $\pi^0,\pi^1,\ldots,\pi^T$ constructed by moving at each round $t$, the element $o_t$ to the first position of the permutation. Notice that $\pi^0,\pi^1,\ldots,\pi^T$ is a feasible solution for both $\mathrm{MoveToFront}$ and $\mathrm{Fractional}-\mathrm{MTF}$. The first key step towards the proof of Lemma~\ref{l:relax} is that
\[\dkt(\pi^t,\pi^{t-1}) + \dkt(\pi^t,\pi^t_{\mathrm{Opt}}) - 
\dkt(\pi^{t-1},\pi^{t}_{\mathrm{Opt}}) \leq 2\cdot \pi^t_{\mathrm{Opt}}(R_t)
\]
To understand the above inequality, let $k_t$ be the position of $o_t$ in permutation $\pi^{t-1}$. Out of the $k_t - 1$ elements on the right of $o_t$ in permutation $\pi^{t-1}$, let $Left_t$ ($Right_t$) denote the elements that are on the left (right) of $o_t$ in permutation $\pi^{t-1}_{\mathrm{Opt}}$. It is not hard to see that $\pi^t_{\mathrm{Opt}}(R_t) \geq |Left_t|$, $\dkt(\pi^t,\pi^{t-1}) = |Left_t| + |Right_t|$ and  
$\dkt(\pi^t,\pi^t_{\mathrm{Opt}}) - 
\dkt(\pi^{t-1},\pi^{t}_{\mathrm{Opt}}) = |Left_t| - |Right_t|$. Using the fact that $\dkt(\pi^{t},\pi^{t}_{\mathrm{Opt}}) - \dkt(\pi^{t-1},\pi^{t}_{\mathrm{Opt}}) \leq \dkt(\pi^{t}_{\mathrm{Opt}},\pi^{t-1}_{\mathrm{Opt}})$ and the previous inequality we get,
\[\dkt(\pi^t,\pi^{t-1}) + \dkt(\pi^t,\pi^t_{\mathrm{Opt}})-
\dkt(\pi^{t-1},\pi^{t-1}_{\mathrm{Opt}})
\leq 
 2 \cdot \pi^t_{\mathrm{Opt}}(R_t) + \dkt(\pi^{t}_{\mathrm{Opt}},\pi^{t-1}_{\mathrm{Opt}})
\]
\noindent and by a telescopic sum we get $\sum_{t=1}^T \dkt(\pi^t,\pi^{t-1}) \leq 2\cdot \mathrm{OPT}_{\DSSC}$. The proof follows by the fact that $\dfr(\pi^t,\pi^{t-1}) \leq 2\cdot  \dkt(\pi^t,\pi^{t-1})$.
\end{proof}

As already mentioned, our main technical contribution is the design of \textit{rounding schemes} converting the optimal solution, $A^1,\ldots,A^T$, of $\mathrm{Fractional- MTF}$ into a sequence of permutations $\pi^1,\ldots,\pi^T$. This is done so as to bound the moving cost of our algorithms by the moving cost $\sum_{t=1}^T \dfr(A^t,A^{t-1})$. We then separately bound the covering cost, $\sum_{t=1}^T \pi^t(R_t)$ by showing that always an element of $R_t$ lies on the first positions of $\pi^t$.

The main technical challenge in the design of our rounding schemes is ensure to that the moving cost of our solutions $\sum_{t=1}^T\dkt(\pi^t,\pi^{t-1})$ is approximately bounded by the moving cost $\sum_{t=1}^T\dfr(A^t,A^{t-1})$. Despite the fact that the connection between 
doubly stochastic matrices and permutations is quite well-studied
and there are various rounding schemes converting doubly stochastic matrices to probability distributions on permutations (such as the Birkhoff–von Neumann decomposition or the schemes of \cite{BGK10,SW11,BBFT20,FKKSV20}), using such schemes in a \textit{black-box} manner does not provide any kind of positive results for $\DSSC$. For example consider the case where $A^1 = \dots = A^T$ and thus $\sum_{t=1}^T \dfr(A^t,A^{t-1}) = \dfr(A^1,A^0)$. In case a randomized rounding scheme is applied \textit{independently to each $A^t$}, there always exists a positive probability that $\pi^t \neq \pi^{t-1}$ and thus the moving cost will far exceed $\dfr(A^1,A^0)$ as $T$ grows. The latter reveals the need for \textit{coupled rounding schemes} that convert the overall sequence of matrices $A^1,\ldots,A^T$ to a sequence of permutations $\pi^1,\ldots,\pi^T$. Such a rounding scheme is presented in Algorithm~\ref{alg:rand_rounding} and constitutes the back-bone of our approximation algorithm for requests of general cardinality.

\begin{algorithm}[ht]
  \caption{A Randomized Algorithm for $\DSSC$}\label{alg:rand_rounding}
  \textbf{Input:} A sequence of requests $R_1,\ldots,R_T$ and an initial permutation of the elmenents $\pi^0$.\\
  \textbf{Output:} A sequence of permutations $\pi^1,\ldots,\pi^T$.

 \begin{algorithmic}[1]
 \STATE Find the optimal solution $A^0=\pi^0,A^1,\ldots,A^T$ for $\mathrm{Fractional-MTF}$. 

     \FOR {each element $e \in U$}
        \STATE Select $\alpha_e$ uniformly at random in $[0,1]$.
        \ENDFOR

 \FOR {$t=1 \ldots T$ }
              
     \FOR {all elements $e \in U$}
                \STATE $I_e^t := \mathrm{argmin}_{1\leq i \leq n} \{ \log n \cdot \sum_{s=1}^i 
                A_{es}^t \geq \alpha_e\}$.
        \ENDFOR
            
            \STATE $\pi^t:=$ sort elements according to $I_e^t$ with ties being broken lexicographically.
\ENDFOR
  \end{algorithmic}
\end{algorithm}
The rounding scheme described in Algorithm~\ref{alg:rand_rounding}, imposes correlation between the different time-steps by simply requiring that each element $e$ selects $\alpha_e$ once and for all and by breaking ties lexicographically (any consistent tie-breaking rule would also work). In Lemma~\ref{l:rand_moving_cost} of Section~\ref{s:rand}, we show that no matter the sequence of doubly stochastic matrices, the rounding scheme of Algorithm~\ref{alg:rand_rounding} produces a sequence of permutations with overall moving cost at most $4\log^2 n$ the moving cost of the matrix-sequence\footnote{By omitting 
the $\log n$-multiplication step of Step~$7$, one could establish that the moving cost of the produced permutations is at most $4$ times the moving cost of the matrix-sequence, however omitting the $\log n$ multiplication could lead in prohibitively high covering cost.} and thus establishes that the overall moving cost of Algorithm~\ref{alg:rand_rounding} is bounded by $4\log^2 n \cdot \mathrm{OPT}_{\DSSC}$. The $\log n$ multiplication in Step~$7$ serves as a \textit{probability amplifier} ensuring that at least one element of $R_t$ lies in the relatively first positions of $\pi^t$ and permits us to approximately bound the covering cost $\sum_{t=1}^T \mathbb{E}\left[\pi^t(R_t) \right]$ by the covering cost of the optimal solution for $\DSSC$,
$\sum_{t=1}^T \pi_\mathrm{Opt}^t(R_t)$.
\begin{theorem}\label{t:rand}
Algorithm~\ref{alg:rand_rounding} is a $O(\log^2 n)$-approximation algorithm for $\DSSC$.
\end{theorem}

\noindent Despite the fact that in Step~$7$ of Algorithm~\ref{alg:rand_rounding}, we multiply the entries of $A^t$ with $\log n$ the overall guarantee is $O(\log^2 n)$. At a first glance the latter seems quite strange but admits a rather natural explanation. For most of the positions $i$, the probability that an element $e$ admits index $I_e^t = i$ is roughly $\log n \cdot A_{ei}^t$, but due to the fact each index $j \leq i$ is on expectation  selected by $\log n$ other elements, the expected position of $e$ in the produced permutation is roughly $ \log^2 n$ times the expected value of
$\mathrm{argmin}_{1\leq i \leq n} \{ \sum_{s=1}^i A_{es}^t \geq \alpha_e\}$. This phenomenon relates with the elegant fitting argument given in \cite{FLT04} to prove that the greedy algorithm is $4$-approximation for the original \textit{Min-Sum Set Cover} (which is tight unless $\mathrm{P}= \mathrm{NP}$). 
The latter makes us conjecture that the tight inapproximability bound for $\DSSC$ is $\Omega(\log^2 n)$ for requests of general cardinality. 

Motivated by the $r$-approximation LP-based algorithm for instances of 
$\mathrm{Set-Cover}$ in which elements belong in at most $r$ sets, we examine whether the $O(\log^2 n)$ for $\DSSC$
can be ameliorated in case of $r$-bounded request sequences. Interestingly, the simple \textit{greedy rounding} scheme (described\footnote{Step~$3$ of Algorithm~\ref{alg:greedy_rounding} is well-defined since $|R_t| \leq r$ and $\sum_{e \in R_t}A_{e1}^t = 1$.} in Algorithm~\ref{alg:greedy_rounding}) provides such a $O(r^2)$-approximation algorithm.
\begin{algorithm}[H]
  \caption{A Greedy-Rounding Algorithm for $\DSSC$ for $r$-Bounded Sequences.}\label{alg:greedy_rounding}
  \textbf{Input:} A request sequence $R_1,\ldots,R_T$ with $|R_t| \leq r$ and an initial permutation $\pi^0$.\\
  \textbf{Output:} A sequence of permutations $\pi^1,\ldots,\pi^T$.

 \begin{algorithmic}[1]
 \STATE Find the optimal solution $A^0=\pi^0,A^1,\ldots,A^T$ for $\mathrm{Fractional-MTF}$.

 \FOR {$t=1 \ldots T$ }
              
        \STATE $\pi^t:=$ in $\pi^{t-1}$, move to the first position an element $e \in R_t$
        such that $A_{e1}^t \geq 1/r$
\ENDFOR
  \end{algorithmic}
\end{algorithm}

\noindent The $O(r^2)$-approximation guarantee of Algorithm~\ref{alg:greedy_moving} is formally stated and proven in Theorem~\ref{t:greedy}. The main technical challenge is that we cannot directly compare the moving cost of Algorithm~\ref{alg:greedy_rounding} with $\sum_{t=1}^T \dfr(A^t,A^{t-1})$ and thus we deploy a two-step detour.

In the first step (Lemma~\ref{l:upper_bound_r}), we prove the existence of a 
sequence of doubly stochastic matrices $\hat{A}^0 =\pi^0,\hat{A}^1,\ldots,\hat{A}^T$ for which each $\hat{A}^t$ satisfies that $\textbf{(i)}$ its entries of  are \textbf{multiples of $1/r$}, $\textbf{(ii)}$ $\hat{A}^t_{e_t 1} \geq 1/r$ where $e_t$ is the element that Algorithm~\ref{alg:greedy_rounding} moves to the first position at round $t$, and $\textbf{(iii)}$ the sequence
$\hat{A}^0 =\pi^0,\hat{A}^1,\ldots,\hat{A}^T$ admits moving cost at most 
$\sum_{t=1}^T \mathrm{d}_{\mathrm{FR}}(A^t,A^{t-1})$. In order to establish the existence of such a sequence, we construct an appropriate linear program (see Definition~\ref{def:lp_2}) based on the elements 
that Algorithm~\ref{alg:greedy_rounding} moves to the first position at each round and prove that it admits an optimal solution with values being multiples of $1/r$. To do the latter, we
relate the linear program of Definition~\ref{def:lp_2} with a fractional version of the $k$-$\mathrm{Paging}$ \cite{BBN12} problem and based on the optimal eviction policy
(\textit{evict the page appearing the furthest in the future}), we design an algorithm producing optimal solutions for the LP with values being multiple of $1/r$.

In the second step (Lemma~\ref{l:r_integral}), we show that for any sequence $\hat{A}^0 =\pi^0,\hat{A}^1,\ldots,\hat{A}^T$ satisfying properties~\textbf{(i)} and~\textbf{(ii)}, the moving cost of Algorithm~\ref{alg:greedy_rounding} is at most $O(r^2) \cdot \sum_{t=1}^T \dfr(\hat{A}^t,\hat{A}^{t-1})$. The latter is achieved through the use of an appropriate \textit{potential function} based on a generalization of Kendall-Tau distance to  
doubly stochastic matrices with entries being multiples
of $1/r$ (see Definition~\ref{d:frac_KT}).
\begin{theorem}\label{t:greedy}
Algorithm~\ref{alg:greedy_rounding}
is a $O(r^2)$-approximation algorithm for $\DSSC$.
\end{theorem}
\noindent In Section~\ref{s:rand}~and~\ref{s:greedy} we provide the basic steps and ideas in the proof of Theorem~\ref{t:rand}~and~\ref{t:greedy} respectively.

\section{Proof of Theorem~\ref{t:rand}}\label{s:rand}
The basic step towards the proof of Theorem~\ref{t:rand} is Lemma~\ref{l:rand_moving_cost}, establishing the fact that once two doubly stochastic matrices are given as input to the randomized rounding of Algorithm~\ref{alg:rand_rounding}, the expected distance of the produced permutations is approximately bounded by the distance of the respective doubly stochastic matrices.
\begin{lemma}\label{l:rand_moving_cost}
Let the doubly stochastic matrices $A,B$ given as input to the
rounding scheme of
Algorithm~\ref{alg:rand_rounding}. Then for the produced permutations
$\pi^A,\pi^B$, 
$
\mathbb{E}\left[ \dkt(\pi^A,\pi^B) \right] \leq 4\log^2 n \cdot \dfr(A,B)$.
\end{lemma}

\noindent Before exhibiting the proof of Lemma~\ref{l:rand_moving_cost} we introduce the notion of \textit{neighboring matrices}.
\begin{definition}(Neighboring stochastic matrices)
The stochastic matrices $A,B$ are neighboring if and only if they differ in exactly two entries lying on the same row and on consecutive columns.
\end{definition}
\begin{example}
Let $A = 
\begin{pmatrix}
1 & 0 & 0 \\
0 & 1 & 0 \\
0 & 0 & 1
\end{pmatrix}$, $B = 
\begin{pmatrix}
1/2 & 1/2 & 0 \\
0 & 1 & 0 \\
0 & 0 & 1
\end{pmatrix}$ and $C = 
\begin{pmatrix}
1 & 0 & 0 \\
1 & 0 & 0 \\
0 & 0 & 1
\end{pmatrix}$. The pair of matrices
$(A,B)$ and $(A,C)$ are neighboring while $(B,C)$ are not.
\end{example}

\noindent Any doubly stochastic matrix $A$ can be converted to another doubly stochastic matrix $B$ through an intermediate sequence of neighboring stochastic matrices all of which are \textit{almost doubly stochastic}
and their overall moving cost equals $\mathrm{d}_{\mathrm{FR}}(A,B)$.
\begin{claim}\label{c:neighboring_sequence}
Given the doubly stochastic matrices $A,B$, there exists a finite sequence of stochastic matrices, $A^0,\ldots,A^T$ such that
\begin{enumerate}
    \item $A^0= A$ and $A^T = B$.
    
    \item $A^t$ and $A^{t-1}$ are neighboring.
    
    \item the column-sum is bounded by $2$, $\sum_{e \in U} A_{ei}^t \leq 2$ for all $1\leq i \leq n$.
    
    \item $\sum_{t=1}^T \dfr(A^t, A^{t-1}) = \dfr(A,B)$.
\end{enumerate}
\end{claim}
\begin{proof}[Proof Sketch of Claim~\ref{c:neighboring_sequence}]
Let $f_{ij}^e$ denotes the optimal solution of the linear program of Definition~\ref{d:distance_lp} defining the FootRule distance $\dfr(A,B)$.  
In case $A\neq B$, there exist elements $e_1 , e_2$ and indices $i<j$ such that $f_{i\ell(i)}^{e_1} > 0$ and $f_{j\ell(j)}^{e_2} > 0$ with $\ell(i) >= j$ and $\ell(j) <=i$.\\

\noindent Let $\epsilon = \min(f_{i\ell(i)}^{e_1}, f_{j\ell(j)}^{e_2})$ and consider the sequence of the $|i-j|$ matrices produced by moving $\epsilon$ amount of mass in row $e_1$ from column $i$ to column $j$. Then consider the sequence of the $|i-j|$ matrices produced by moving $\epsilon$ amount of mass in the row $e_2$ from column $j$ to column $i$.\\

\noindent In the overall sequence of $2|i-j|$ stochastic matrices, two consecutive matrices are \textit{neighboring}. Furthermore the column-sum of the matrices does not exceed $1 + \epsilon \leq 2$ and the final matrix $A'$ of the sequence is doubly stochastic. Moreover by the fact that $t(i) \geq j$ and $t(j) \leq i$ we get that the overall moving cost of the sequence equals $\dfr(A,A')$ and that 
$\dfr(A,B) = \dfr(A,A') + \dfr(A',B)$. Applying the same argument inductively, until we reach matrix $B$, proves Claim~\ref{c:neighboring_sequence}.
\end{proof}

\begin{example}
Let the doubly stochastic matrices $A = 
\begin{pmatrix}
1 & 0 & 0 \\
0 & 1 & 0 \\
0 & 0 & 1
\end{pmatrix}$, $B = 
\begin{pmatrix}
0 & 0 & 1 \\
1/2 & 1/2 & 0 \\
1/2 & 1/2 & 0
\end{pmatrix}$. $A$ can be converted to $B$ with the following sequence neighboring stochastic matrices, 
\smallskip
$\begin{pmatrix}
0 & 1 & 0 \\
0 & 1 & 0 \\
0 & 0 & 1
\end{pmatrix}$,
$\begin{pmatrix}
0 & 0 & 1 \\
0 & 1 & 0 \\
0 & 0 & 1
\end{pmatrix}$,
$\begin{pmatrix}
0 & 0 & 1 \\
0 & 1 & 0 \\
0 & 1 & 0
\end{pmatrix}$,
$\begin{pmatrix}
0 & 0 & 1 \\
1/2 & 1/2 & 0 \\
0 & 1 & 0
\end{pmatrix}$,
$\begin{pmatrix}
0 & 0 & 1 \\
1/2 & 1/2 & 0 \\
1/2 & 1/2 & 0
\end{pmatrix}$.
\smallskip
Notice that the above sequence satisfies all the $4$ requirements of Claim~\ref{c:neighboring_sequence}.
\end{example}

\noindent The notion of neighboring matrices is rather helpful since Lemma~\ref{l:rand_moving_cost} admits a fairly simple proof in case $A,B$ are neighboring stochastic matrices (notice
that the rounding scheme of
Algorithm~\ref{alg:rand_rounding} is well-defined even for stochastic matrices). The latter is formally stated and proven in Lemma~\ref{l:neigh} and is the main technical claim of the section.

\begin{lemma}\label{l:neigh}
Let $\pi^A,\pi^B$ the permutations produced by the rounding scheme of
Algorithm~\ref{alg:rand_rounding} (given as input) the stochastic matrices $A,B$ that 
\textbf{i)} are neighboring \textbf{ii)} their column-sum is bounded by $2$, then
$\mathbb E[\dkt(\pi^A,\pi^B)] \leq 4 \log^2 n \cdot \dfr(A,B)$
\end{lemma}

\begin{proof}[Proof of Lemma~\ref{l:neigh}]
Since $A,B$ are neighboring there
exists exactly two consecutive entries for which $A,B$ differ, denoted as $(e^\ast,i^\ast)$ and $(e^\ast,i^\ast+1)$. Let $\epsilon := A_{e^\ast i^\ast} - B_{e^\ast i^\ast }$, by the  Definition~\ref{d:distance_lp} of FootRule distance, we get that $\dfr(A,B) = |\epsilon|$. Without loss of generality we consider $\epsilon >0$ (the case $\epsilon <0$ symmetrically follows). We also denote with $O_i$ the set of elements $O_i := \{e \neq e^\star \text{ such that } I_e^A = i\}$ and with $I_e^A,I_e^B$ the indices in Step~$6$ of Algorithm~\ref{alg:rand_rounding}.\\

\noindent Since $A,B$ are neighboring, the $e$-th row of $A$ and the $e$-th row of $B$ are identical for all $e \neq e^\star$. As a result, $I_e^A = I_e^B$ for all $e \neq e^\star$. Furthermore the neighboring property implies that
even for $e^\ast$, $\sum_{s=1}^i A_{e^\star s} = \sum_{s=1}^i B_{e^\star s}$ for all $i \neq i^\star$ and thus $\Pr \left[ I_{e^\star}^A=i \wedge I_{e^\star}^B = j \right] = 0$ for $(i,j) \neq (i^\star,i^\star+1)$. Now notice that
\begin{align*}
\Pr\left[ I_{e^\star}^A=i^\star, I_{e^\star}^B = i^\star + 1 \right] &\leq 
\Pr\left [ \log n \cdot \sum_{s=1}^{i^\star} B_{e^\star s}
\leq \alpha_e \leq \log n \cdot \sum_{s=1}^{i^\star} A_{e^\star s} \right] \\
&\leq 
\log n \cdot \left( A_{e^\star i ^\star} -  B_{e^\star i ^\star}\right)  = \log n \cdot \epsilon
\end{align*}

\noindent Notice also that in case $I_{e^\star}^A  = I_{e^\star}^B$, $\mathrm{d}_{\mathrm{KT}}(\pi_A,\pi_B) = 0$. This is due to the fact that in such a case $I_e^A = I_e^B$ for all $e \in U$ and the fact that ties are broken lexicographically. As a result,
\begin{align*}
 \mathbb{E}\left[\mathrm{d}_{\mathrm{KT}}(\pi_A,\pi_B) \right] &= \Pr\left[ I_{e^\star}^A \neq I_{e^\star}^B\right] \cdot \mathbb{E}\left[\mathrm{d}_{\mathrm{KT}}(\pi_A,\pi_B)|~ I_{e^\star}^A \neq I_{e^\star}^B\right]\\
&=\Pr[ I_{e^\star}^A=i^\star, I_{e^\star}^B = i^\star + 1 ] \cdot \mathbb{E}\left[\mathrm{d}_{\mathrm{KT}}(\pi_A,\pi_B)|~ I_{e^\star}^A=i^\star, I_{e^\star}^B =i^\star + 1\right]\\ 
&\leq \epsilon \log n \cdot \left( \mathbb{E}\left[|O_{i^\star}|\right] + \mathbb{E}\left[|O_{i^\star + 1}|\right] \right)
\end{align*}
where the last inequality follows by the fact that once $I_{e^\ast}^A = i^\ast$ and $I_{e^\ast}^B = i^\ast + 1$, the element $e^\ast$ can move at most by $|O_{i^\ast}| + |O_{i^\ast+1}|$ positions and the fact that $I_{e^\ast}^A,I_{e^\ast}^B$ and $|O_{i^\ast}|,|O_{i^\ast+1}|$ are independent random variables.\\

\noindent We complete the proof we providing a bound on $\mathbb{E}\left[|O_i|\right]$.  Notice that for $e \in U/ \{e^\ast\}$,
$$ \Pr[ e \in O_{i}] \leq \Pr \left [ \log n \sum_{s=1}^{i-1} A_{es} \leq \alpha_e \leq \log n \sum_{s=1}^{i} A_{es}\right] \leq \log n \cdot A_{ei} $$
which implies that $\mathbb{E}\left[|O_i|\right] \leq \log n \sum_{e \neq e^\star} A_{ei} \leq 2 \log n$. Finally we overall get,
$$\mathbb{E}\left[\mathrm{d}_{\mathrm{KT}}(\pi_A,\pi_B) \right] \leq
4\log^2 n \cdot \dfr(A,B)$$
\end{proof}

\noindent The proof of Lemma~\ref{l:rand_moving_cost} easily follows by Claim~\ref{c:neighboring_sequence} and Lemma~\ref{l:neigh}. 

\begin{proof}[Proof of Lemma~\ref{l:rand_moving_cost}] Given the doubly stochastic matrices $A,B$, let the sequence  
$A = A^0,A^1,\ldots,A^T = B$ of neighboring stochastic matrices ensured by  Claim~\ref{c:neighboring_sequence}. Now let $\pi^0,\pi^1,\ldots,\pi^T$ the sequence of permutations that the randomized rounding of Algorithm~\ref{alg:rand_rounding} produces given as input the sequence $A = A^0,A^1,\ldots,A^T = B$. Notice that,
\[\mathbb{E}\left[\mathrm{d}_{\mathrm{KT}}(\pi^A,\pi^B)\right] \leq
\sum_{t=1}^t \mathbb{E}\left[\mathrm{d}_{\mathrm{KT}}(\pi^t,\pi^{t-1})\right]
\leq 4\log^2 n \cdot 
\sum_{t=1}^T\dfr(A^t,A^{t-1})
=  4\log^2 n \cdot \dfr(A,B)\]
where the first inequality follows by the triangle inequality, the second by Lemma~\ref{l:neigh}
and the last equality by Case $4$ of Claim~\ref{c:neighboring_sequence}.
\end{proof}

\noindent We conclude the section with the proof of Theorem~\ref{t:rand}.
\begin{proof}[Proof of Theorem~\ref{t:rand}]
By Lemma~\ref{l:rand_moving_cost} and Lemma~\ref{l:relax},
$$ \sum_{t=1}^T \mathbb{E}\left[\dkt(\pi^t,\pi^{t-1})\right] \leq 
4\log^2 n \cdot \sum_{t=1}^T \dfr(A^t,A^{t-1}) \leq 4\log^2 n \cdot \mathrm{OPT}_{\DSSC}
$$

\noindent Up next we bound the expected covering cost $\sum_{t=1}^T \mathbb{E}\left[\pi^t(R_t)\right]$. Notice that since $\sum_{e \in R_t} A_{e1}^t = 1$, the only elements that can have index $I_e^t = 1$ are the elements $e \in R_t$. As a result, in case there exists some $e$ at round $t$ with $I_e^t = 1$ then $\pi^t(R_t) = 1$.
\begin{eqnarray*}
\mathbb{E}\left[\pi^t(R_t)\right] &\leq& 1 + n \cdot \Pr \left[ I_e^t >1 \text{ for all } e\in R_t \right]\\
&\leq& 1 + n \cdot \Pi_{e \in R_t} \left(1 - \log n \cdot A_{e1}^t\right)\\
&\leq& 1 + n \cdot e^{- \log n \cdot \sum_{e \in R_t}A_{e1}^t}= 2 \cdot  \pi_{\mathrm{Opt}}^t(R_t)
\end{eqnarray*}
where the last inequality follows due to the fact that $\sum_{e \in R_t} A_{e1}^t = 1$ and $\pi_{\mathrm{Opt}}^t(R_t) \geq 1$.
\end{proof}

\section{Proof of Theorem~\ref{t:greedy}}\label{s:greedy}
 \noindent In this section we present the basic steps towards the proof of Theorem~\ref{t:greedy}. We remind that $|R_t| \leq r$ and we denote with $e_t$ the element that Algorithm~\ref{alg:greedy_rounding} moves in the fist position at round $t$. As already mentioned, the proof is structured in two 
 different steps.
 \smallskip
 \begin{enumerate}
     \item We prove the existence of a sequence of doubly stochastic matrices $\hat{A}^0 = \pi^0,\hat{A}^1,\ldots,\hat{A}^T$ such that \textbf{(i)} the entries of each $\hat{A}^t$ are multiples of $1/r$, \textbf{(ii)} each $\hat{A}^t$ admits $1/r$ mass for element $e_t$ in first position $(\hat{A}^t_{e_t 1} \geq 1/r)$ and
     \textbf{(iii)}
     $\sum_{t=1}^T \dfr(\hat{A}^t,\hat{A}^{t-1}) \leq \sum_{t=1}^T \dfr(A^t,A^{t-1})$.\smallskip
     \item We use properties $\textbf{(i)}$ and $\textbf{(ii)}$ to prove that the moving cost of Algorithm~\ref{alg:greedy_rounding} is roughly upper bounded by $\Theta(r^2) \cdot \sum_{t=1}^T \dfr(\hat{A}^t,\hat{A}^{t-1})$.
 \end{enumerate}

\noindent We start with the construction of the sequence $\hat{A}^0=\pi^0,\hat{A}^1,\ldots,\hat{A}^T$.
\begin{definition}\label{def:lp_2}
For the sequence of elements $e_1,\ldots,e_T \in U$ (the elements that Algorithm~\ref{alg:greedy_rounding} moves to the fist position at each round), consider the following linear program,
\begin{equation*}
    \begin{array}{ll@{}ll}
        \text min & \displaystyle \sum_{t=1}^T \dfr(\hat{A}^t,\hat{A}^{t-1})&\\
        \text{ s.t } & \displaystyle \sum_{i=1}^{n} \hat{A}_{ei}^t = 1 ~~~~~~~~~~\text{for all } e \in U \text{ and } t = 1,\ldots, T  &&\\
        & \displaystyle \sum_{e \in U} \hat{A}_{ei}^t = 1
        ~~~~~~~~~~\text{for all } i = 1, \ldots, n \text{ and } t = 1, \ldots, T  &&\\
        & \displaystyle \hat{A}_{e_t1}^t \geq 1/r ~~~~~~~~~~~\text{for all } t = 1, \ldots, T&\\
        &\displaystyle \hat{A}^0 = \pi^0 &\\
        & \displaystyle\hat{A}_{ei}^t\geq 0~~~~~~~~~~~~~~~\text{for all } e \in U, ~i = 1,\ldots, n  \text{ and }t = 1,\ldots, T&\\
    \end{array}
\end{equation*}
\end{definition}
\noindent The sequence $\hat{A^0}=\pi^0,\ldots,\hat{A}^T$ is defined as the optimal solution of the LP in Definition~\ref{def:lp_2} with the entries of each $\hat{A}^t$ being \textbf{multiples of $1/r$}. The existence of such an optimal solution is established in Lemma~\ref{l:upper_bound_r}.   

\begin{lemma}\label{l:upper_bound_r}
There exists an optimal solution $\hat{A} = \pi^0,\hat{A}^1,\ldots,\hat{A}^T$ for the linear program of Definition~\ref{l:upper_bound_r} such that entries of each $\hat{A}^t$ are multiples of $1/r$.
\end{lemma}
The proof of Lemma~\ref{l:upper_bound_r} is one of the main technical contributions of this work. Due to lack of space its proof is deferred to  the full version of the paper. We remark that the \textit{semi-integrality property}, that Lemma~\ref{l:upper_bound_r} states, is not due to the properties of the LP's polytope and in fact there are simple instances in which the optimal extreme points do not satisfy it. We establish Lemma~\ref{l:upper_bound_r} via the design of an optimal algorithm for the LP of Definition~\ref{def:lp_2} (Algorithm~\ref{alg:greedy_moving}) that always produces solutions with entries being multiples of $1/r$. Up next we describe in brief the idea behind Algorithm~\ref{alg:greedy_moving}.

Given the matrix $\hat{A}^{t-1}$, Algorithm~\ref{alg:greedy_moving} construct $\hat{A}^t$ as follows. 
At first it moves $1/r$ mass from the left-most entry $(e_t,j)$ with
$\hat{A}^{t-1}_{e_t j} \geq 1/r$ to the entry $(e_t,1)$. At this point the third constraint of the LP in Definition~\ref{def:lp_2} is satisfied but the column-stochasticity constraints are violated (the first column admits mass $1+1/r$ and the $j$-th column admits mass $1-1/r$). Algorithm~\ref{alg:greedy_moving} inductively restores column-stochasticity from left to right. At step $i$, all the columns on the left of $i$ are restored and the violations concern the column $i$ and $j$ ($i$'s mass is $1+1/r$ and $j$'s mass is $1-1/r$). Now Algorithm~\ref{alg:greedy_moving} must move a total of $1/r$ mass from column $i$ to column $i+1$. In case there exists an element $e$ with total amount of mass greater than $2/r$, Algorithm~\ref{alg:greedy_rounding} moves the $1/r$ mass from the entry $(e,i)$ to the entry $(e,i+1)$. The reason is that even if $e = e_{t'}$
at some future round $t'$, the third constraint only requires $1/r$ mass. In case there is no such element, Algorithm~\ref{alg:greedy_moving} moves the $1/r$ mass from the element appearing the furthest in the sequence $\{e_t,\ldots,e_T\}$. The latter is in accordance with the optimal eviction policy for $\mathrm{k}-\mathrm{Paging}$ which at each round evicts the page appearing furthest in the future \cite{BBN12}. The optimality of Algorithm~\ref{alg:greedy_moving} is established in Lemma~\ref{l:optimality_Greedy_Moving} and the fact that produced solution admits values being $1/r$ is inductively established.

To this end, we can show that all of the desired properties of the sequence $\hat{A} = \pi^0,\hat{A}^1,\ldots,\hat{A}^T$
are satisfied. Property~$\textbf{(i)}$ is established by Lemma~\ref{l:upper_bound_r}. Property~$\textbf{(ii)}$ is enforced by the constraint $\hat{A}_{e_t 1}^t \geq 1/r$. Now for Property~$\textbf{(iii)}$, notice that by the definition of Algorithm~\ref{alg:greedy_rounding}, $A_{e_t 1}^t \geq 1/r$. As a result, the sequence $A^0=\pi^0,A^1,\ldots,A^T$ is feasible for the linear program of Definition~\ref{def:lp_2} and thus $\sum_{t=1}^T \dfr(\hat{A}^t,\hat{A}^{t-1}) \leq \sum_{t=1}^T \dfr(A^t,A^{t-1})$.

\begin{lemma}\label{l:r_integral}
Let $\pi^0,\pi^1,\ldots,\pi^T$ the permutations produced by Algorithm~\ref{alg:greedy_rounding} and $e_1,\ldots,e_T$ the elements that Algorithm~\ref{alg:greedy_rounding} moves to the first position at each round $t$. For any sequence of doubly stochastic matrices $\hat{A}^0 = \pi^0,\hat{A^1},\ldots,\hat{A^T}$ for which Property~\textbf{(i)} and~Property~\textbf{(ii)} are satisfied, $\sum_{t=1}^T \dkt(\pi^t,\pi^{t-1}) \leq 2r^2 \cdot  \sum_{t=1}^T \dfr(\hat{A}^t,\hat{A}^{t-1}) + r \cdot T$.
\end{lemma}
\noindent The proof of Theorem~\ref{t:greedy} directly follows by Lemma~\ref{l:upper_bound_r} and~\ref{l:r_integral}. In Section~\ref{sub:greedy_1} we present the basic steps for of Lemma~\ref{l:upper_bound_r}.
\subsection{Proof of Lemma~\ref{l:upper_bound_r}}\label{sub:greedy_2}
We prove the existence of an optimal solution $\hat{A^0} = \pi^0,\hat{A^1},\ldots,\hat{A}^T$ for the linear program of Definition~\ref{def:lp_2} for which the entries of each matrix $\hat{A}^t$ are multiples of $1/r$ though the design of an optimal greedy algorithm illustrated in Algorithm~\ref{alg:greedy_moving}.  

The fact that Algorithm~\ref{alg:greedy_moving} produces a solution with entries that multiples of $1/r$ easily follows. Algorithm~\ref{alg:greedy_moving} starts with an integral doubly stochastic matrices ($\hat{A}^0 = \pi^0$) and always moves $1/r$ mass from entry to entry. The optimality of Algorithm~\ref{alg:greedy_moving} is established in Lemma~\ref{l:optimality_Greedy_Moving} the proof of which is presented in the next section since it is quite technically complicated. However the basic idea of the algorithms is very intuitive, once $\hat{A}^{t-1}_{e_t} = 0$ Algorithm~\ref{alg:greedy_moving} moves $1/r$ mass of $e_t$ from its leftmost position (with mass greaer than $1/r$), denoted as $\mathrm{Pos}$ of Step~$5$. At this point of time, Algorithm~\ref{alg:greedy_moving} has violated the column-stochasticity constraints, $1+1/r$ for the first column and $1-1/r$ for the $\mathrm{Pos}$-th column and Algorithm~\ref{alg:greedy_moving} must move at total of $1/r$ mass from the first position to next positions until $1/r$ mass reaches the $\mathrm{Pos}$ position and column-stochasticity is restored (Step~$8$). Once Algorithm~\ref{alg:greedy_moving} detects an element with aggregated mass (until position $j$) $\geq 2/r$, it can safely move $1/r$ of each mass to position $j+1$ since even if this element appears at some point in the future only $1/r$ is necessary to satisfy the constraint $A_{e_t 1}^t \geq 1/r$ and thus the rest is redundant (Step~$11$). In case such an element does not exist, Algorithm~\ref{alg:greedy_moving} moves the (useful) $1/r$ mass of the element appearing the furthest in the remaining sequence $\{e_t,\ldots,e_{T}\}$, which is exactly the same optimal \textit{eviction policy} that the well-studied $k-\mathrm{Paging}$ suggests. 

\begin{lemma}\label{l:optimality_Greedy_Moving} Algorithm~\ref{alg:greedy_moving} produces an optimal solution $\hat{A}^0=\pi^0,\hat{A}^1,\ldots,\hat{A}^T$ for the linear program of Definition~\ref{def:lp_2} while the entries of each $\hat{A}^t$ are multiples of $1/r$.
\end{lemma}

\begin{algorithm}[t]
  \caption{An Optimal Greedy Algorithm for the LP of Definition~\ref{def:lp_2} 
  }\label{alg:greedy_moving}
  \textbf{Input:} The initial permutation $\pi^0$ and the sequence of elements $e_1,\ldots,e_T \in U$\\
  \textbf{Output:} An optimal solution of a linear program of Definition~\ref{def:lp_2} where the  entries of $\hat{A}^t$ are multiples of $1/r$.

 \begin{algorithmic}[1]
        \STATE Initially $\hat{A}^0 \leftarrow \pi_0$
    
        \FOR{all rounds $t = 1$ \text{ to } $T$}
        
        \STATE $ \hat{A}^t \leftarrow \hat{A}^{t-1}$
  
        \IF {$\hat{A}^t_{e_t 1 } <  1/r $} 
            \STATE \emph{//Move $1/r$ mass of $e_t$ to the first position}
            
            \STATE $~~~~~~~~\mathrm{Pos} \leftarrow \text{argmin}_{1 \leq i \leq n} \{A^t_{ei}\geq 1/r\}$
        
            \STATE $~~~~~~~~\hat{A}^t_{e1} \leftarrow \hat{A}^t_{e1} + 1/r, \hat{A}^t_{e\mathrm{Pos}} \leftarrow \hat{A}^t_{e\mathrm{Pos}} - 1/r$
            
            \STATE \emph{//Restore the column-stochasticity constraints from left to right}
            \FOR{$j = 1$ \text{ to } $\mathrm{Pos} - 1$}

            \IF{ there exists $e \in U$ with $\sum_{s=1}^j \hat{A}_{es}^t \geq 2/r$ and $\hat{A}_{es}^t \geq 1/r$}
    
                \STATE \emph{//Move $1/r$ of its (redundant) mass to the next position}
                \STATE $~~~~\hat{A}_{ej}^t \leftarrow \hat{A}_{ej}^t - 1/r$, $\hat{A}_{ej}^t \leftarrow \hat{A}_{ej}^t + 1/r$   
    
            \ELSE
            
            \STATE \emph{//Move the $1/r$ mass, of the element appearing furthest in the future, to the next position}
            \STATE $~~~~~~~e^\star \in U \leftarrow$  the element with $\hat{A}_{e^\star j}^t = 1/r$ furthest in $\{e_{t+1},\ldots,e_T\}$
            
            \STATE $~~~~~~~\hat{A}_{e^\star j}^t \leftarrow \hat{A}_{e^\star j}^t - 1/r$, $\hat{A}_{e^\star j}^t \leftarrow \hat{A}_{e^\star j}^t + 1/r$
            \ENDIF    
            \ENDFOR            
            
        \ENDIF 
    
        \ENDFOR
        
        \RETURN $\hat{A}_1,\ldots,\hat{A}_T$
  \end{algorithmic}
\end{algorithm}
\subsection{Proof of Lemma~\ref{l:r_integral}}\label{sub:greedy_1}
In order to prove Lemma~\ref{l:r_integral}, we make use of an appropriate potential function that can be viewed as an extension of the Kendall-Tau distance (see Definition~\ref{d:KT}) to doubly stochastic matrices with entries being multiples of $1/r$.
\begin{definition}[\textbf{$r$-Index}]
The $r$-index of an element $e \in U$ in the doubly stochastic matrix $A$, $I_e^A:=\mathrm{argmin}\{1\leq i \leq n:~ \sum_{s=1}^i A_{es} \geq 1/r\}$
\end{definition}

\begin{definition}[\textbf{Fractional Kendall-Tau Distance}]\label{d:frac_KT} Given the doubly stochastic matrices $A,B$, a pair of elements $(e, e') \in U \times U$ is inverted if and only if one of the following condition holds,
\begin{enumerate}
    \item $I_e^A > I_{e'}^A$ and $I_e^B < I_{e'}^B$.
    \item $I_e^A < I_{e'}^A$ and $I_e^B > I_{e'}^B$.
    \item $I_e^A = I_{e'}^A$ and $I_e^B \neq I_{e'}^B$.
    \item $I_e^A \neq I_{e'}^A$ and $I_e^B = I_{e'}^B$.
\end{enumerate}
The fractional Kendall-Tau distance between two doubly stochastic matrices $A,B$, denoted as $\dkt(A,B)$, is the number of inverted pairs of elements. 
\end{definition}
\noindent Notice that in case of $0-1$ doubly stochastic matrices the Fractional Kendall-Tau distance of Definition~\ref{d:frac_KT} coincides with the Kendall-Tau distance of Definition~\ref{d:KT}. 

\begin{claim}\label{c:metric}
Fractional Kendall-Tau Distance satisfies the triangle inequality, $\dkt(A,B) \leq \dkt(A,C) + \dkt(C,B)$.
\end{claim}

\begin{proof}[Proof of Claim~\ref{c:metric}]
Let $X_{ee'}^{AB} = 1$ if $(e,e')$ is inverted pair for the matrices $A,B$ and $0$ otherwise (respectively for $X_{ee'}^{AC},X_{ee'}^{BC}$). By a short case study one can show that once $X_{ee'}^{AB} = 1$ then
$X_{ee'}^{AC} + X_{ee'}^{BC} \geq 1$ which directly implies Claim~\ref{c:metric}.
\end{proof}

\noindent In the case of doubly stochastic matrices with their entries being multiples of $1/r$, Fractional Kendall-Tau distance relates to FootRule distance of Definition~\ref{d:distance_lp}.
\begin{lemma}\label{l:equivalence}
    Let the doubly stochastic matrices $A,B$ with entries that are multiples of $1/r$. Then $\dkt ( A, B ) \leq 2r^2 \cdot  \dfr( A, B )$.
\end{lemma}

\begin{proof}[Proof of Lemma~\ref{l:equivalence}]
    We construct a doubly stochastic matrix $A'$ for which the following properties hold,
    \begin{enumerate}
        \item The entries of $A'$ are multiples of $\frac{1}{r}$.
        \item $\dfr(A,B) = \dfr(A, A') + \dfr(A',B)$.
        \item $\mathrm{d}_{\mathrm{KT}}( A, A' ) \leq 2r^2 \cdot  \dfr(A, A')$.
    \end{enumerate}
    
\noindent Once the above properties are established, Lemma~\ref{l:equivalence} follows by repeating the same construction until matrix $B$ is reached
and by using the fact that the \textit{fractional Kendall-Tau distance} of Definition~\ref{d:frac_KT}
satisfies the triangle inequality.\\

\noindent Before proceeding with the construction of $A'$, we present the following corollary that follows by an easy exchange argument.
\begin{corollary}\label{c:flow}
    Let the stochastic matrices $A,B$ with entries multiples of $1/r$, the values $f_{ij}^e$ of the optimal solution in the linear program of Definition~\ref{d:distance_lp} (the min-cost transportation problem defining the FootRule distance $\mathrm{d}_{\mathrm{FR}}(A,B)$) are multiples of $1/r$.
    \end{corollary}

\noindent In order to construct the matrix $A'$ satisfying the Properties~$1$-$3$, we consider three different classes of the entries $(e,i)$. In particular, we call an entry $(e,i)$.
    \begin{enumerate}
        \item \textit{right} if and only if $f_{ij}^e > 0$ for some $j > i$.
        \item \textit{left} if and only if $f_{ij}^e > 0 $ for some $j < i$.
        \item \textit{neutral} if and only if $f_{ij}^e = 0$ for all $j \neq i$.
    \end{enumerate}
    
\noindent Note that the above classes do not form a partition of the entries since an entry $(e,i)$ can be both \textit{left} and \textit{right} at the same time.    

\begin{corollary}\label{cor:2}
    Given two doubly stochastic matrices $A \neq B$, there exist entries $(e,i)$ and $(e',j)$ such that
    \begin{enumerate}
        \item $j > i$
        \item the entry $(e,i)$ is right
        \item the entry $(e',j)$ is left
        \item the entry $(\alpha,\ell)$ is \textit{neutral} for all $\alpha \in U$ and $\ell \in \{i+1,j-1\}$
    \end{enumerate}
    \end{corollary}
    
    We construct the matrix $A'$ from matrix $A$ as follows. Consider two entries $(e,i)$ and $(e',j)$ with the properties that Corollary~\ref{cor:2} illustrates. The doubly stochastic matrix $A'$ is constructed by moving $1/r$ mass from entry $(e,i)$ to entry $(e,j)$ and by moving $1/r$ mass from entry $(e',j)$ to entry $(e',i)$. More formally,
    \begin{equation*}
        A'_{\alpha \ell} = 
        \begin{cases} 
            A_{\alpha \ell} - \frac{1}{r} & \text{ if } (\alpha,\ell)=(e,i) \\
            A_{\alpha \ell} - \frac{1}{r} & \text{ if } (\alpha,\ell)=(e',j)\\
            A_{\alpha k} + \frac{1}{r} & \text{ if } (\alpha,\ell)=(e',i) \\
            A_{\alpha \ell} + \frac{1}{r} & \text{ if } (\alpha,\ell)=(e,j)\\
            A_{\alpha \ell} & \text{ otherwise }
        \end{cases}
    \end{equation*}
\noindent Up next we establish the fact that $\dfr(A,B) = \dfr(A,A') + \dfr(A',B)$.
\begin{claim}
$\dfr(A',A) = 2|j-i|/r$ and 
$\dfr(A',B) = \dfr(A,B) - 2|j-i|/r$.    
\end{claim}
\begin{proof}
The fact that $\dfr(A',A) = 2|j-i|/r$ is trivial. We thus focus on showing that 
$\dfr(A',B) = \dfr(A,B) - 2|j-i|/r$.\\

\noindent Since $(e,i)$ is \textit{right}, there exists an index $\ell(i) >i$ such that $f_{i \ell(i)}^e > 0$. Moreover $f_{i \ell(i)}^e \geq 1/r$ since $f_{i \ell(i)}^e$ is multiple of $1/r$. Notice that $\ell(i) \neq \ell$ for $\ell \in \{i+1,j-1\}$ since all the entries $(\alpha,\ell)$ are \textit{neutral} (otherwise $\sum_{\alpha \in U} B_{\alpha \ell } > 1$). As a result, transfering $1/r$ mass from entry $(e,i)$ to entry $(e,j)$ decreases the FootRule distance between $A$ and $B$ by $1/r\cdot|i-j|$ since the \textit{final destination} of the $1/r$ mass is the entry $(e,\ell(i))$ that is on the right of entry $(e,j)$, $\ell(i) \geq j$. The claim follows by applying the exact same argument for $(e',j)$.    
\end{proof}

\noindent We now establish the last property that is $\mathrm{d}_{\mathrm{KT}}(A,A') \leq 2r^2 \cdot \dfr(A,A')$.
    
    \begin{claim}
        $\mathrm{d}_{\mathrm{KT}}( A', B ) \leq 4r \cdot | i - j|$
    \end{claim}
    
    \begin{proof}
        \noindent Notice that apart from $e,e'$, the $r$-index of each element is the same in both $A$ and $A'$ ($I_\alpha^{A} = I_\alpha^{A'}$ for all $\alpha \in U \setminus \{e, e'\}$). As a result, by Definition~\ref{d:frac_KT}, we get that the only inverted pairs can be of the form $(e,\alpha)$ or $(e',\alpha)$.\\
    
        \noindent In case $I_e^A \leq i-1$ then $I_e^A = I_e^{A'}$ and there is no inverted pair of the form $(e,\alpha)$. In case $I_e^A = i$ then $i \leq I_e^{A'} \leq j$ and any element $\alpha$ with $I_\alpha^A =
        I_\alpha^{A'} \in \{1,i-1\}\cup \{j+1,n\}$ cannot form an inverted pair with $e$. As a result, a pair $(e,\alpha)$ can be inverted only if $i \leq I_\alpha^A = I_\alpha^{A'} \leq j$. Since the entries of A are multiples of $1/r$ and $A$ is doubly stochastic, there are at most
        $r$ positive entries at each column of $A$. As a result, there are at most $r \cdot (j-i+1)$ inverted pairs of the form $(e,\alpha)$. With the symmetric argument one can show that there are at most $r\cdot |j-i+1|$ of the form $(e',\alpha)$. Overall there are at most $2r\cdot |j-i+1|$ inverted pairs between $A$ and $A'$ that are less than $4r\cdot |j-i|$ since $j > i$.
    \end{proof}
    \end{proof}

\noindent We conclude the section with Lemma~\ref{l:potential}.
Then Lemma~\ref{l:r_integral} follows by Lemma~\ref{l:potential} 
and~\ref{l:equivalence}.
\begin{lemma}\label{l:potential}
Let $\pi^0,\pi^1,\ldots,\pi^T$ the permutations produced by Algorithm~\ref{alg:greedy_rounding} and $e_1,\ldots,e_T$ the elements that Algorithm~\ref{alg:greedy_rounding} moves to the first position at each round $t$. For any sequence of doubly stochastic matrices $B^0 = \pi^0,B^1,\ldots,B^T$ with $B_{e_t 1}^t \geq 1/r$,
$$\sum_{t=1}^T \dkt(\pi^t,\pi^{t-1}) \leq \sum_{t=1}^T \dkt(B^t,B^{t-1}) + r \cdot T
$$
\end{lemma}
The proof of Lemma~\ref{l:potential} is based on the following two inequalities, $\dkt(\pi^t,\pi^{t-1}) + \dkt(\pi^t,B^t) - \dkt(\pi^{t-1},B^t) \leq r$ and $\dkt(\pi^{t-1},B^t) - \dkt(\pi^{t-1},B^{t-1}) \leq \dkt(B^t,B^{t-1})$. The second inequality follows by the triangle inequality established in Claim~\ref{c:metric}. The first follows by the fact that $I_{e_t}^{B^t} = 1$ and the definition of Fractional Kendall-Tau distance.

\begin{proof}[Proof of Lemma~\ref{l:potential}]
Since $B_{e_t}^t \geq 1/r$, the $r$-index of element $e_t$ in matrix $B^t$ is $1$, $I_{e_t}^{B^{t}} = 1$. We first show that, $$\mathrm{d}_{\mathrm{KT}}\left(\pi^t, \pi^{t-1}\right)
+ \mathrm{d}_{\mathrm{KT}}\left(\pi^t, B^{t}\right)
- \mathrm{d}_{\mathrm{KT}}\left(\pi^{t-1}, B^{t}\right) \leq r$$ 
To simplify notation let $k_t$ the position of $e_t$ in $\pi^{t-1}$. Notice that $\mathrm{d}_{\mathrm{KT}}\left(\pi^t, \pi^{t-1}\right) = k_t -1$. Out of the $k_t - 1$ elements lying on the left of $e_t$ in $\pi^{t-1}$ there are most $r-1$ elements $\alpha$ with $I_{\alpha}^{B^t} = 1$ (these elements must admit $B_{\alpha 1}^t \geq 1/r$). The rest of the $k_t - 1$ elements admit $r$-index $I_{\alpha}^{B^t} \geq 2$ and thus form inverted pairs with $e_t$ when considering $\pi^{t-1}$ and $B^t$. When $e_t$ moves to the first positions (permutation $\pi^t$) these inverted pairs are deactivated ($I_{e_t}^{B^t} = 1$) and new inverted pairs are created between $e_t$ and $\alpha$ with $I_{\alpha}^{B^t} = 1$, but these new inverted pairs are at most $r$ (for any element $\alpha$ with $I_{\alpha}^{B^t}$, $B^t_{\alpha} \geq 1/r$). Also notice no additional inverted pairs $(e,\alpha)$ (with $e \neq e_t$) are created since the order between all the other elements is the same in $\pi^t$ and $\pi^{t-1}$. Overall,
$$\underbrace{\mathrm{d}_{\mathrm{KT}}\left(\pi^t, \pi^{t-1}\right)}_{ k_t - 1}
+ \underbrace{\mathrm{d}_{\mathrm{KT}}\left(\pi^t, B^{t}\right)
- \mathrm{d}_{\mathrm{KT}}\left(\pi^{t-1}, B^{t}\right)}_{\leq -k_t + 1 + r} \leq r$$

\noindent Combining the above inequality with $\mathrm{d}_{\mathrm{KT}}\left(\pi^{t-1}, B^{t}\right)
- \mathrm{d}_{\mathrm{KT}}\left(\pi^{t-1}, B^{t-1}\right) \leq \mathrm{d}_{\mathrm{KT}}\left(B^t, B^{t-1}\right)$ which follows from the triangle inequality we get,
$$\mathrm{d}_{\mathrm{KT}}\left(\pi^t, \pi^{t-1}\right)
+ \mathrm{d}_{\mathrm{KT}}\left(\pi^t, B^{t}\right)
- \mathrm{d}_{\mathrm{KT}}\left(\pi^{t-1}, B^{t-1}\right) \leq \mathrm{d}_{\mathrm{KT}}\left(A^{t}, B^{t-1}\right) +  r.$$
Finally a telescopic sum gives $\sum_{t=1}^T \mathrm{d}_{\mathrm{KT}}\left(\pi^t, \pi^{t-1}\right) \leq \sum_{t=1}^T \mathrm{d}_{\mathrm{KT}}\left(B^t, B^{t-1}\right)
+ r\cdot T + \mathrm{d}_{\mathrm{KT}}(\pi^0,B^0) - \mathrm{d}_{\mathrm{KT}}(\pi^T,B^T)
$ where $\mathrm{d}_{\mathrm{KT}}(\pi^0,B^0) = 0$.
\end{proof}

\section{Concluding Remarks}
In this work we examine the polynomial-time approximability of Multistage Min-Sum Set Cover. We present $\Omega(\log n)$ and $\Omega(r)$ inapproximability results for general and $r$-bounded request sequences, while we respectively provide $O(\log^2 n)$ and $O(r^2)$ polynomial-time approximation algorithms. Closing this gap is an interesting question that our work leaves open. Another interesting research direction concerns the competitive ratio in the online version of Dynamic Min-Sum Set Cover. 
\cite{FLPS20} provides an $\Omega(r)$ lower bound and a $\Theta\left(r^{3/2}\sqrt{n}\right)$-competitive online algorithm for $r$-bounded sequences. Designing online algorithms for a relaxation of the problem (such as the $\mathrm{Fractional}-\mathrm{MTF}$) and using the rounding schemes that this work suggests may be a fruitful approach towards closing this gap.

\bibliography{references}

\end{document}